\documentclass[reqno]{amsart}

\usepackage{booktabs}
\usepackage{IEEEtrantools}
\usepackage{amssymb,latexsym,amsfonts,amsmath}
\usepackage{mathrsfs}
\usepackage{graphicx}
\usepackage{dsfont}
\usepackage{tikz}
\usetikzlibrary{calc,shapes,arrows}

\newcommand{\intcc}[1]{\ensuremath{{\left[#1\right]}}}

\usepackage{algorithm}
\usepackage{algorithmic}

\usepackage{xspace}

\newtheorem{theorem}{Theorem}[section]
\newtheorem{lemma}[theorem]{Lemma}

\newtheorem{definition}[theorem]{Definition}

\newtheorem{remark}[theorem]{Remark}

\numberwithin{equation}{section}

\newcommand{\R}{{\mathbb{R}}}

\newcommand{\N}{{\mathbb{N}}}

\newcommand{\Let}{:=}
\newcommand{\EE}{\mathds{E}}
\newcommand{\PP}{\mathds{P}}

\usepackage{fancyhdr}

\newenvironment{nouppercase}{%
	\renewcommand{\uppercasenonmath}[1]{}}{}

\begin{document}

\begin{abstract}
This work proposes a compositional data-driven technique for the construction of finite Markov decision processes (MDPs) for large-scale stochastic networks with \emph{unknown} mathematical models. Our proposed framework leverages dissipativity properties of subsystems and their finite MDPs using a notion of \emph{stochastic storage functions} (SStF). In our data-driven scheme, we first build an SStF between each unknown subsystem and its data-driven finite MDP with a certified probabilistic confidence. We then derive dissipativity-type compositional conditions to construct a \emph{stochastic bisimulation function} (SBF) between an interconnected network and its finite MDP using data-driven SStF of subsystems. Accordingly, we formally quantify the probabilistic distance between trajectories of an unknown large-scale stochastic network and those of its finite MDP with a guaranteed confidence. We illustrate the efficacy of our data-driven results over a room temperature network composing $100$ rooms with unknown models.
\end{abstract}

\title{{\LARGE MDP Abstractions from Data: Large-Scale Stochastic Networks}}

\author{{\bf {\large Abolfazl Lavaei,~\emph{{\small Senior Member,~IEEE}}}}\\{\normalfont School of Computing, Newcastle University, United Kingdom}}

\pagestyle{fancy}
\lhead{}
\rhead{}
  \fancyhead[OL]{Abolfazl Lavaei}

  \fancyhead[EL]{MDP Abstractions from Data: Large-Scale Stochastic Networks}
  \rhead{\thepage}
 \cfoot{}
 
\begin{nouppercase}
	\maketitle
\end{nouppercase}

\section{Introduction}

Providing a formal analysis framework for large-scale stochastic networks to fulfill complex logic properties is generally very challenging. This is particularly due to  (i) dealing with uncountable state/input sets with large dimensions, (ii) stochastic nature of dynamics, (iii) complex logic requirements, and (iv) lack of closed-form mathematical models in many real-world applications. To mitigate the aforesaid difficulties, one rewarding solution is to approximate the original (concrete) system by a finite MDP as a finite-state model. By establishing a similarity relation between each concrete system and its finite MDP using a notion of \emph{stochastic simulation functions}, the probabilistic mismatch between two systems can be quantified within a guaranteed error bound.

There have been numerous studies, conducted in the past two decades, on the abstraction-based analysis of stochastic systems. Existing results encompass construction of (in)finite abstractions for \emph{stochastic} dynamical systems with continuous state sets~\cite{APLS08,julius2009approximations,zamani2014symbolic,lahijanian2015formal}. However, the main bottleneck of those techniques is \emph{curse of dimensionality} problem due to discretizing state and input sets. \emph{Compositional techniques} for constructing finite abstractions have then been proposed to alleviate the underlying state-explosion problem: one can build a finite abstraction for a large-dimensional system using finite abstractions of smaller subsystems~\cite{hahn2013compositional,lavaei2022scalable,nejati2021compositional,lavaei2022dissipativity,nejati2020compositional,lavaei2019compositional,ref4}. 

Although the above-mentioned studies on constructing finite abstractions are
comprehensive, unfortunately, they require knowing the mathematical model of the system. Accordingly, one cannot leverage those techniques for many practical scenarios with unknown models. Although \emph{identification techniques} have been proposed to learn approximate models of unknown systems, obtaining a precise model is computationally very burdensome~\cite[and references herein]{Hou2013model}). In addition, even
if a model can be identified using system identification
techniques, the relation between the identified model and its finite abstraction should be still constructed. Consequently,
the computational complexity exists in two levels of identifying the model and establishing the similarity relation. In this work, we develop a \emph{direct} data-driven scheme, without performing any system identification, and construct finite
abstractions together with their associated similarity relations by directly gathering data from trajectories of unknown concrete systems.

The original contribution here is to propose a compositional data-driven technique for constructing finite MDPs for large-scale stochastic control networks with unknown mathematical models. We leverage dissipativity properties of subsystems and their finite MDPs using a notion of \emph{stochastic storage functions} (SStF). In our data-driven scheme, we recast conditions of SStF as a robust optimization program (ROP). By gathering samples from trajectories of each unknown subsystem, we then provide a scenario optimization program (SOP) for the original ROP. By quantifying the closeness between the optimal values of SOP and ROP, we build an SStF between each unknown subsystem and its data-driven finite MDP with a guaranteed probabilistic confidence. We then derive a dissipativity-type compositional condition to construct \emph{stochastic bisimulation functions} (SBF), between an interconnected network and its finite MDP, using data-driven SStF of subsystems. Eventually, we quantify the probabilistic closeness between trajectories of an unknown interconnected network and its finite MDP with a guaranteed confidence level. We demonstrate the efficacy of our proposed data-driven results over a room temperature network composing $100$ rooms with unknown models.

There has been a limited number of work on data-driven construction of symbolic models (in deterministic setting) and finite MDPs (in stochastic setting). Existing results include: data-driven abstraction of monotone systems with disturbances\cite{makdesi2021efficient}, data-driven construction of symbolic abstractions via a probably approximately correct (PAC) approach~\cite{devonport2021symbolic}; data-driven construction of finite abstractions for verification of unknown systems~\cite{coppola2022data}; data-driven construction of symbolic models for incrementally input-to-state stable systems~\cite{Lavaei_LCSS22_2}; and data-driven construction of finite MDPs for incrementally input-to-state stable systems~\cite{Lavaei_LCSS22_1}. In comparison, we propose here a compositional data-driven framework using dissipativity approach for constructing finite MDPs for large-scale interconnected networks, whereas the results
in~\cite{makdesi2021efficient,devonport2021symbolic,coppola2022data,Lavaei_LCSS22_2,Lavaei_LCSS22_1} are all tailored to monolithic systems. As a result, the approaches in~\cite{makdesi2021efficient,devonport2021symbolic,coppola2022data,Lavaei_LCSS22_2,Lavaei_LCSS22_1} suffer from the sample complexity problem and are not useful in practice when dealing with high-dimensional systems. In addition, the works~\cite{makdesi2021efficient,devonport2021symbolic,coppola2022data,Lavaei_LCSS22_2} construct \emph{symbolic} abstractions from unknown \emph{deterministic} systems, whereas we develop here a data-driven technique for building \emph{finite MDPs} for \emph{stochastic} systems which is more challenging given the stochastic nature of unknown dynamics.

\section{Discrete-Time Stochastic Control Systems}\label{Sec: dt-NDS}

\subsection{Notation and Preliminaries}

In this work, $\mathbb{R},\mathbb{R}^+$, and $\mathbb{R}^+_0$, represent sets of real, positive, and non-negative real numbers, respectively. Symbols $\mathbb{N} := \{0,1,2,...\}$ and $\mathbb{N}^+=\{1,2,...\}$ denote, respectively, sets of non-negative and positive integers. A column vector, given $N$ vectors $x_i \in \mathbb{R}^{n_i}$, is represented by $x=[x_1;\dots;x_N]$. Given a set $X$, its power set is denoted by $2^X$. We denote the minimum and maximum eigenvalues of a symmetric matrix $P$, respectively, by $\lambda_{\min}(P)$ and $\lambda_{\max}(P)$. Given any scalar $a\in\mathbb R$ and vector $x\in\mathbb{R}^{n}$, $\vert a\vert$ and $\Vert x\Vert$ represent, respectively, the absolute value and the Euclidean norm. For a matrix $P\in\mathbb R^{m\times n}$, $\|P\| := \sqrt{\lambda_{\max}(P^\top P)}$. We denote the supremum of a function $f:\mathbb N\rightarrow\mathbb{R}^n$ by $\Vert f\Vert_{\infty} \Let \text{(ess)sup}\{\Vert f(k)\Vert,k\geq 0\}$. Given a system $\Lambda$ and a property $\varphi$, $\Lambda \vDash \varphi$ denotes that $\Lambda$ fulfills $\varphi$.

Given a probability space $(\Omega, \mathcal{F}_\Omega, \PP_\Omega)$, with $\Omega$ being a sample space, $\mathcal{F}_\Omega$ a sigma-algebra on $\Omega$, and $\PP_\Omega$ a probability measure, $N$-Cartesian product set of $\Omega$ and its associated product measure are denoted, respectively, by $\Omega^N$ and $\PP^N$.
A set $X$ is Borel, denoted by $\mathcal B(X)$, if it is homeomorphic to a Borel subset of a Polish space, \textit{i.e.}, a separable and metrizable space.

\subsection{Discrete-Time Stochastic Control Systems}\label{Sec: dt-SCS}
Here, we first formally define discrete-time stochastic control systems as the following.
\begin{definition}
	A discrete-time stochastic control system (dt-SCS) is characterized by
	\begin{align}\label{EQ:13}
		\Lambda=(X,U,D,\varsigma,f),
	\end{align}
	where:
	\begin{itemize}
		\item $X\subseteq \mathbb R^n$ is a Borel state set;
		\item $U = \{\nu_1,\nu_2,\dots,\nu_m\}$, with $\nu_i \in \mathbb R^{\bar m}, i \in\{1,\dots,m\}$, is a discrete input set;
		\item $D\subseteq \mathbb R^p$ is a Borel disturbance set;
		\item $\varsigma$ is a sequence of independent-and-identically distributed
		(i.i.d.) random variables from the sample space $\Omega$ to a set
		$\mathcal{H}_{\varsigma}$, \emph{i.e.} $\varsigma:=\{\varsigma(k):\Omega\rightarrow \mathcal H_{\varsigma},\,\,k\in\N\}$;
		\item $f:X\times U\times D\times\mathcal{H}_{\varsigma}\rightarrow X$ is a transition map, which is assumed to be unknown.
	\end{itemize}
\end{definition}
The evolution of dt-SCS can be described by 
\begin{align}\label{EQ:14}
	\Lambda\!:x(k+1)=f(x(k),\nu(k),d(k),\varsigma(k)),\quad k\in\mathbb N, 
\end{align} 
for any $x\in X$, $\nu(\cdot):\Omega\rightarrow U$, and $d(\cdot):\Omega\rightarrow D$. The \emph{state trajectory} of $\Lambda$ under $\nu(\cdot), d(\cdot)$ starting from $x(0)= x_0$ is denoted by $x_{x_0\nu d}\!\!: \Omega \times \mathbb N \rightarrow X$. 

Since the ultimate objective is to construct a finite MDP for an \emph{interconnected} dt-SCS, we consider dt-SCS in~\eqref{EQ:14} as a \emph{subsystem} and present another definition for \emph{interconnected} dt-SCS without disturbances $d$ as a composition of individual dt-SCS with disturbances $d$. 

\begin{definition}\label{Def:2}
	Consider $M\in\mathbb N^+$ dt-SCS $\Lambda_i=(X_i,U_i,D_i,f_i,\varsigma_i)$, $i\in \{1,\dots, M\}$, with a matrix $\mathcal M$ as a coupling among them. An interconnection of $\Lambda_i$ is characterized as
	$\Lambda=(X,U,f, \varsigma)$, represented by
	$\mathcal{I}(\Lambda_1,\ldots,\Lambda_M)$, where $X:=\prod_{i=1}^{M}X_i$, $U:=\prod_{i=1}^{M}U_i$, $f:=[f_1;\dots;f_{M}]$, and $\varsigma:=[\varsigma_1;\dots;\varsigma_M]$, such that:
	\begin{align}\label{Eq:41}
		\Big[d_{1};\cdots;d_{M}\Big]=\mathcal M~\!\Big[x_1;\cdots;x_M\Big]\!.
	\end{align}
	Such an interconnected dt-SCS is described by
	\begin{equation}\label{Eq:5}
		\Lambda\!:x(k+1)\!=\!f(x(k),\nu(k),\varsigma(k)),  \text{with}~ f: X \times U \times \mathcal{H}_{\varsigma}\rightarrow X.
	\end{equation}
\end{definition}

An interconnected dt-SCS $\Lambda$ is schematically depicted in Fig.~\ref{Fig1}.

\begin{figure} 
	\begin{center}
		\includegraphics[width=0.42\linewidth]{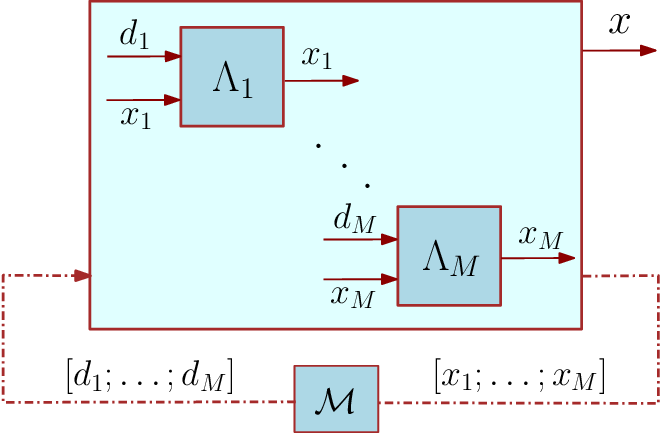} 
		\caption{Interconnected dt-SCS $\mathcal{I}(\Lambda_1,\ldots,\Lambda_M)$.} 
		\label{Fig1}
	\end{center}
\end{figure}

\subsection{Finite Markov Decision Processes}\label{MDPs}

Here, we construct finite MDPs as finite-state approximations of dt-SCS. To this end, we first partition state and disturbance sets as $X = \cup_i \mathsf X_i$ and $D = \cup_i \mathsf D_i$,
and then pick representative points $\hat x_i\in \mathsf X_i$ and $\hat d_i\in \mathsf D_i$ within those partitions sets as finite states and disturbances. 

The dt-SCS in \eqref{EQ:13} can be \emph{equivalently} considered as a continuous-space MDP $\Lambda=(X,U,D,\mathsf{T}_{\mathsf x})$~\cite{kallenberg1997foundations}, 
with $\mathsf{T}_{\mathsf x}\!:\mathcal B(X)\times X\times U\times D\rightarrow[0,1]$ being a conditional stochastic kernel that assigns to any $x \in X, \nu\in U, d\in D$, a probability measure $\mathsf{T}_{\mathsf x}(\cdot \,\big |\, x,\nu,d)$
such that for any set $\mathcal X \in \mathcal B(X)$:\vspace{-0.2cm} 
\begin{align*}
	\PP \big\{x(k+1)\in \mathcal X&\,\big |\, x(k),\nu(k),d(k)\big\}= \int_{\mathcal X} \mathsf{T}_{\mathsf x} (\mathsf{d}x(k+1)\,\big |\,x(k),\nu(k),d(k)).
\end{align*}
One can \emph{uniquely} determine the conditional stochastic kernel $\mathsf{T}_{\mathsf x}$ using $(\varsigma,f)$~\cite{kallenberg1997foundations}. In the next definition, we formalize the construction of finite MDPs.

\begin{definition}
	Consider a continuous-space MDP $\Lambda=(X,U,D,\mathsf{T}_{\mathsf x})$. The \emph{finite MDP}, constructed from $\Lambda$, is characterized by  $\hat \Lambda\!=\!(\hat X, U,\hat D,\mathsf{\hat T}_{\mathsf x})$, with $\hat X$ and $\hat D$ being discrete state and disturbance sets of~$\hat\Lambda$ and
	\begin{align*}
		\mathsf{\hat T}_{\mathsf x} (x'\big|x,\nu,d) 
		\!=\! \mathsf{T}_{\mathsf x} (\Xi(x')\big|x,\nu,d), \forall x,x' \!\in\! \hat X, \forall \nu\!\in\! U, \forall d\!\in\!D,
	\end{align*}
	where $\Xi:X\rightarrow 2^X$. Equivalently, given a dt-SCS $\Lambda=(X,U,D,\varsigma,f)$, its constructed \emph{finite MDP} can be characterized as~\cite{kallenberg1997foundations}
	\begin{equation*}
		\hat\Lambda =(\hat X, U,\hat D,\varsigma,\hat f),
	\end{equation*}
	where $\hat f\!:\hat X\times U\times \hat D \times\mathcal H_{\varsigma}\rightarrow\hat X$ is a transition function defined as
	\begin{equation}\label{EQ:3}
		\hat f(\hat{x},\nu,\hat d,\varsigma) = \mathcal P_{x}(f(\hat{x},\nu,\hat d,\varsigma)),
	\end{equation}
	and $\mathcal P_x:X\rightarrow \hat X$ is a quantization map with a \emph{state discretization parameter} $\rho$ fulfilling the following inequality: 
	\begin{equation}\label{EQ:4}
		\Vert \mathcal P_x(x)-x\Vert \leq \rho,\quad \forall x\in X. 
	\end{equation}
\end{definition}

\section{Stochastic Storage and Bisimulation Functions}
In this section, we aim at quantifying the probabilistic mismatch between trajectories of an interconnected dt-SCS and its finite MDP using a notion of \emph{stochastic bisimulation functions}, as defined next.

\begin{definition}\label{Def:3}
	Given an interconnected dt-SCS $\Lambda =(X,U,\varsigma,f)$ and its 
	finite MDP $\hat\Lambda =(\hat X,U,\varsigma,\hat f)$, a function $\mathcal V\!:X\times\hat X\to\R_0^+$ is a stochastic bisimulation function (SBF) between $\hat\Lambda$ and $\Lambda$, represented by $\hat\Lambda\cong_{\mathcal{V}}\Lambda$, if
	\begin{subequations}
		\begin{align}\label{EQ:15}
			&\forall x\in X, \forall \hat x\in\hat X\!: \quad\quad\quad\quad\quad\quad\quad\quad\quad\quad\quad\quad\quad\gamma\Vert x - \hat x\Vert^2\leq \mathcal V(x,\hat x),\\\label{EQ:16}
			&\forall x\in X, \forall \hat x\in\hat X\!, \forall\nu\in U\!: \quad\quad\quad\quad\EE \Big [\mathcal V(f(x,\nu,\varsigma),\hat{f}(\hat x,\nu,\varsigma))\,\big |\, x, \hat x, \nu\Big ]\leq\alpha \mathcal V(x,\hat{x})+\varpi,
		\end{align}
	\end{subequations}
	for some $\gamma\in\R^+$, $0< \alpha <1$, and $\varpi \in\R_0^+$, where $\EE$ is the expected value associated to $\varsigma$.
\end{definition}

We now leverage SBF $\mathcal V$ and quantify the probabilistic mismatch between trajectories of an interconnected system and its finite MDP, as in the next theorem~\cite{ref4}. 

\begin{theorem}\label{Thm:4}
	Given an interconnected dt-SCS $\Lambda$ and its 
	finite MDP $\hat\Lambda$, let $\mathcal V$ be an SBF between $\hat\Lambda$ and $\Lambda$. Then the probabilistic closeness between state trajectories of dt-SCS (\emph{i.e.} $x_{x_0\nu}(\cdot)$) and its
	finite MDP (\emph{i.e.} $\hat x_{\hat x_0\nu}(\cdot)$) within a time horizon $\mathcal T\in \mathbb N$ can be quantified as
	\begin{align}\label{delta}
		&\PP\left\{\sup_{0\leq k\leq \mathcal T}\Vert x_{x_0\nu}(k)-\hat x_{\hat x_0\nu}(k)\Vert\geq\varepsilon\,\big |\,x_0,\hat x_0\right\} \le \delta,
	\end{align}
	where
	\begin{align*}
		&\delta:=
		\begin{cases}
			1\!-\!(1\!-\!\frac{\mathcal V(x_0,\hat x_0)}{\gamma\varepsilon^2})(1\!-\!\frac{\varpi}{\gamma\varepsilon^2})^{\mathcal T}, & ~~\text{if}~\gamma\varepsilon^2\!\geq\!\frac{\varpi}{1\!-\!\alpha},\\
			(\frac{\mathcal V(x_0,\hat x_0)}{\gamma\varepsilon^2})\alpha^{\mathcal T}\!+\!(\frac{\varpi}{(1\!-\!\alpha)\gamma\varepsilon^2})(1\!-\!\alpha^{\mathcal T}), & ~~\text{if}~\gamma\varepsilon^2\!<\!\frac{\varpi}{1-\alpha},
		\end{cases}
	\end{align*}
	with $\varepsilon\in \mathbb R^+$ being an arbitrary threshold. If $\varpi=0$ in \eqref{EQ:16}, the closeness guarantee in~\eqref{delta} can be generalized to infinite horizons as 
	\begin{align*}
		\PP\left\{\sup_{0\leq k<\infty}\Vert x_{x_0\nu}(k)\!-\!\hat x_{\hat x_0\nu}(k)\Vert\geq\varepsilon\,\big |\,x_0,\hat x_0\right\} \!\le\! \frac{\mathcal V(x_0,\hat x_0)}{\gamma\varepsilon^2}.
	\end{align*}
\end{theorem}

In general, constructing SBF for large-scale interconnected networks is very expensive (if it is not impossible). To tackle this computational difficulty, we present a notion of \emph{stochastic storage functions} for individual subsystems and propose, in Section~\ref{Sec:Com}, some compositional dissipativity conditions to construct an SBF for an interconnected network using SStF of subsystems.

\begin{definition}\label{Def:31}
	Given a dt-SCS $\Lambda =(X,U,D,\varsigma,f)$ and its 
	finite MDP $\hat\Lambda =(\hat X,U,\hat D,\varsigma,\hat f)$, a function $\mathcal S\!:X\times\hat X\to\R_0^+$ is a stochastic storage function (SStF) between $\hat\Lambda$ and $\Lambda$, represented by $\hat\Lambda\cong_{\mathcal{S}}\Lambda$, if
	\begin{subequations}
		\begin{align}\label{EQ:151}
			&\forall x\in X, \forall \hat x\in\hat X\!: \quad\quad\quad\quad\quad\quad\quad\quad\quad\quad\quad\quad\gamma\Vert x - \hat x\Vert^2\leq \mathcal S(x,\hat x),\\\notag
			&\forall x\in X, \forall \hat x\in\hat X\!, \forall\nu\in U\!, \forall d\in D, \forall\hat d\in\hat D\!: \\\label{EQ:161}
			&\EE \Big [\mathcal S(f(x,\nu,d,\varsigma),\hat{f}(\hat x,\nu,\hat d,\varsigma))\,\big |\, x, \hat x, \nu,d,\hat d\Big ]\leq\alpha\mathcal S(x,\hat{x})+\varpi +\begin{bmatrix}
				d-\hat d\\
				x-\hat x
			\end{bmatrix}^\top
			\underbrace{\begin{bmatrix}
					\mathcal Z^{11}&\mathcal Z^{12}\\
					\mathcal Z^{21}&\mathcal Z^{22}
			\end{bmatrix}}_{\mathcal Z}\!\begin{bmatrix}
				d-\hat d\\
				x-\hat x
			\end{bmatrix}\!\!,
		\end{align}
	\end{subequations}
	for some $\gamma\in\R^+$, $0< \alpha <1$, $\varpi \in\R_0^+$, and a symmetric matrix $\mathcal Z$ with partitions $\mathcal Z^{jj'},$ $j,j'\in\{1,2\}$.
\end{definition}

\section{Data-Driven Construction of SStF}
In our data-driven framework, we consider SStF in the form of $\mathcal{S}(\kappa,x,\hat x)=\sum_{j=1}^{z} {\kappa}_jh_j(x,\hat x)$ with basis functions $h_j(x,\hat x)$ and unknown variables $\kappa=[{\kappa}_{1};\ldots;\kappa_z] \in \mathbb{R}^z$. We now cast conditions~\eqref{EQ:151}-\eqref{EQ:161} of SStF as a robust optimization program (ROP):
\begin{align}\label{ROP}
	&\text{ROP}\!:\!\left\{
	\hspace{-1.5mm}\begin{array}{l}\min\limits_{[\mathcal Q;\psi]} \quad\!\!\!\psi,\\
		\, \text{s.t.} \quad  \,\max_j\Big\{\Upsilon_j(x,\hat x,\nu,d, \hat d, \mathcal Q)\Big\}\leq \psi,~  j\in\{1,2\}, \\ 
		\quad\quad\quad\!\forall x\in X, \forall \hat x\in \hat X, \forall \nu\in U, \forall d\in D, \forall \hat d\in \hat D,\\
		\quad\quad\quad\!\mathcal Q = [\gamma;\alpha;\varpi;\mathcal Z^{11};\mathcal Z^{12};\mathcal Z^{22};{\kappa}_{1};\dots;\kappa_z],\\
		\quad\quad\quad\!\gamma\!\in\!\R^+, \alpha \!\in\!(0,1), \varpi \!\in\!\R_0^+,\mathcal Z^{jj'}\!\!, \psi \!\in\! \mathbb R,\end{array}\right.
\end{align}
where:
\begin{align}\notag
	\Upsilon_1 &= \gamma\Vert x - \hat x\Vert^2 - \mathcal S(\kappa,x,\hat x),\\\label{EQ:11}
	\Upsilon_2 &= \EE \,\Big [\mathcal S(\kappa,f(x,\nu,d,\varsigma),\hat{f}(\hat x,\nu,\hat d,\varsigma)) \,\big |x, \hat x,\nu,d, \hat d\Big ]-\alpha\mathcal S(\kappa,x,\hat{x})-\varpi- \begin{bmatrix}
		d\!-\!\hat d\\
		x\!-\!\hat x
	\end{bmatrix}^\top\!\!
	\begin{bmatrix}
		\mathcal Z^{11}\!\!&\!\!\mathcal Z^{12}\\
		\mathcal Z^{21}\!\!&\!\!\mathcal Z^{22}
	\end{bmatrix}\!\begin{bmatrix}
		d\!-\!\hat d\\
		x\!-\!\hat x
	\end{bmatrix}\!\!.
\end{align}
When $\psi_R^*$, the optimal value of ROP, is less than or equal to zero, it is straightforward to confirm that conditions \eqref{EQ:151}-\eqref{EQ:161} are met.

The ROP in~\eqref{ROP} is not solvable due to unknown maps $f$, $\hat f$ appearing in $\Upsilon_2$. To resolve this difficulty, we collect $N$ i.i.d. data within $X\times D$, denoted by $(\bar x_i,\bar d_i)^{N}_{i=1}$. We now propose a scenario optimization program (SOP), with an optimal value $\psi_N^*$, associated to the original ROP:
\begin{subequations}
	\begin{align}\label{SOP}
		&\text{SOP}_N\!\!:\!\left\{
		\hspace{-1.5mm}\begin{array}{l}\min\limits_{[\mathcal Q;\psi]} \,\,\,\,\,\!\!\psi,\\
			\, \text{s.t.}  \quad \,\max_j\!\Big\{\Upsilon_j(\bar x_i, \hat x,\nu,\bar d_i,\hat d,\mathcal Q)\Big\}\leq \psi, ~\! j\!\in\!\{1,2\}, \\
			\quad \quad\quad \!\!\forall \bar x_i\in X,\forall \bar d_i\in D,\forall i\!\in\! \{1,\ldots, N\},\\
			\quad \quad\quad \!\!\forall \hat x\in \hat X,\forall \hat d\in \hat D,\forall \nu\in U,\\
			\quad\quad\quad\!\!\mathcal Q = [\gamma;\alpha;\varpi;\mathcal Z^{11};\mathcal Z^{12};\mathcal Z^{22};{\kappa}_{1};\dots;\kappa_z],\\
			\quad\quad\quad\!\gamma\!\in\!\R^+, \alpha \!\in\!(0,1), \varpi \!\in\!\R_0^+,\mathcal Z^{jj'}\!\!, \psi \!\in\! \mathbb R.\end{array}\right.
	\end{align}
	We can now replace the unknown function $f(\bar x_i, \nu, \bar d_i, \varsigma)$ in $\Upsilon_2$ by observing the one-step transition of dt-SCS starting from $\bar x_i$ under $\nu$ and $\bar d_i$. Regarding $\hat f(\hat x, \nu, \hat d, \varsigma)$ in $\Upsilon_2$, we begin by initializing the unknown model at $\hat x$ under $\nu$ and $\hat d$ to compute $f(\hat x, \nu, \hat d, \varsigma)$. With a state discretization parameter $\rho$ in place, we then compute $\hat f(\hat x, \nu, \hat d, \varsigma)$ as the point nearest to $f(\hat x, \nu, \hat d, \varsigma)$, where condition \eqref{EQ:4} is satisfied.
	
	By proposing SOP~\eqref{SOP}, the problem of unknown maps $f,\hat f$ in ROP~\eqref{ROP} got solved. However, the proposed SOP in~\eqref{SOP} is not still tractable since there is no closed-form solution for computing the expected value in $\Upsilon_2$. To resolve this issue, we propose another version of SOP, denoted by SOP$_\varsigma$, by computing the expected value using its empirical approximation: 
	\begin{align}\label{SCP1}
		&\text{SOP}_\varsigma\!\!:\!\left\{
		\hspace{-2mm}\begin{array}{l}\min\limits_{[\mathcal Q;\psi]} \,\,\,\,\,\!\!\!\!\!\psi,\\
			\, \text{s.t.}  \quad \,\!\!\!\max\Big\{\Upsilon_{1}(\bar x_i, \hat x,\nu,\bar d_i,\hat d,\mathcal Q),\bar \Upsilon_{2}(\bar x_i,\hat x,\nu,\bar d_i,\hat d,\mathcal Q)\Big\}\leq\psi, \\
			\quad \quad\quad\!\!\!\!\forall \bar x_i\in X,\forall \bar d_i\in D,\forall i\!\in\! \{1,\ldots, N\},\\
			\quad \quad\quad\!\!\!\!\forall \hat x\in \hat X,\forall \hat d\in \hat D,\forall \nu\in U,\\
			\quad \quad\quad\!\!\!\!\mathcal Q = [\gamma;\alpha;\varpi;\mathcal Z^{11};\mathcal Z^{12};\mathcal Z^{22};{\kappa}_{1};\dots;\kappa_z],\\
			\quad\quad\quad\!\!\!\gamma\!\in\!\R^+, \alpha \!\in\!(0,1), \varpi \!\in\!\R_0^+,\mathcal Z^{jj'}\!\!, \kappa_i,\psi \!\in\! \mathbb R,\end{array}\right.
	\end{align}
\end{subequations}
with
\begin{align*}
	\bar\Upsilon_{2}&\!=\!\frac{1}{L}\sum_{q=1}^{L}\mathcal{S}(\kappa,f(\bar x_i,\nu,\bar d_i,\varsigma_{q}),\hat f(\hat x,\nu,\hat d,\varsigma_{q}))\!-\! \alpha\mathcal{S}(\kappa,x_i,\hat{x})-\varpi + \mu- \begin{bmatrix}
		d-\hat d\\
		x-\hat x
	\end{bmatrix}^\top\!
	\begin{bmatrix}
		\mathcal Z^{11}&\mathcal Z^{12}\\
		\mathcal Z^{21}&\mathcal Z^{22}
	\end{bmatrix}\!\begin{bmatrix}
		d-\hat d\\
		x-\hat x
	\end{bmatrix}\!\!,
\end{align*}
where $\mu\in \mathbb{R}_0^+$ and $L\in\mathbb{N}_{0}^+$ are the approximation error and required number of realizations, respectively. We denote the optimal value of $\text{SOP}_\varsigma$ by $\psi_\varsigma^*$. 

We now leverage Chebyshev's inequality~\cite{saw1984chebyshev} to construct a relation between solutions of SOP$_\varsigma$  and SOP$_N$  with a guaranteed confidence level $\beta_{1}\in (0,1]$. 

\begin{lemma}\label{Lem:2}
	Let ${\mathcal{S}}$ be a feasible solution for SOP$_\varsigma$ in~\eqref{SCP1}. For a desired confidence level $\beta_{1}\in (0,1]$ and an approximation error $\mu\in \mathbb{R}_0^+$, one has
	\begin{align*}
		\PP\Big\{{\mathcal{S}}(\kappa,x,\hat x) \models\text{SOP}_N\Big\}\geq 1- \beta_{1},
	\end{align*}
	provided that $L \geq \frac{\mathcal{C}}{\beta_{1}\mu^2}$, where $\text{Var}\big[\mathcal S(\kappa,f(x,\nu,d,\varsigma),\hat{f}(\hat x,\nu,\hat d,\varsigma))\big] \! \!\leq \!\!\mathcal{C},$ $\forall x \in X, \forall \hat x \in \hat X, \forall \nu \in U, \forall d \in D, \forall \hat d \in \hat D$.
\end{lemma}

\begin{proof}
	Using Chebyshev's inequality~\cite{saw1984chebyshev}, the closeness between the expected value in $\Upsilon_2$ and its empirical mean in $\bar \Upsilon_2$ can be quantified as, $\forall x \in X, \forall \hat x \in \hat X, \forall \nu \in U, \forall d \in D, \forall \hat d \in \hat D$,
	\begin{align*}
		\PP&\Big\{~\!\Big|\EE \Big  [\mathcal S(\kappa,f(x,\nu,d,\varsigma),\hat{f}(\hat x,\nu,\hat d,\varsigma))\,\big |\, x, \hat x,\nu,d,\hat d \Big ]\\ &-\frac{1}{L}\sum_{q=1}^{L}\mathcal{S}(\kappa,f(\bar x_i,\nu,\bar d_i,\varsigma_{q}),\hat f(\hat x,\nu,\hat d,\varsigma_{q}))\Big|\!\leq\!\mu\Big\}\!\geq\! 1\!-\!\frac{\sigma^2}{\mu^2},
	\end{align*}
	for any $\mu\in \mathbb{R}_0^+$, where
	\begin{align*}
		&\sigma^2 = \text{Var}\Big[\frac{1}{L}\sum_{q=1}^{L}\mathcal{S}(\kappa,f(\bar x_i,\nu,\bar d_i,\varsigma_{q}),\hat f(\hat x,\nu,\hat d,\varsigma_{q}))\Big]\!.
	\end{align*}
	Since $\text{Var}\big[\mathcal S(\kappa,f(x,\nu,d,\varsigma),\hat{f}(\hat x,\nu,\hat d,\varsigma))\big]\leq \mathcal{C},\forall x \in X, \forall \hat x \in \hat X, \forall \nu \in U, \forall d \in D, \forall \hat d \in \hat D$, one has $\sigma^2 \leq \frac{\mathcal C}{L}$. Consequently, $\beta_{1}=\frac{\sigma^2}{\mu^2} \leq\frac{\mathcal{C}}{L\mu^2}$. It implies that for $L \geq \frac{\mathcal{C}}{\beta_{1}\mu^2}$: 
	\begin{align*}
		\PP\Big\{{\mathcal{S}}(\kappa,x,\hat x) \models\text{SOP}_N\Big\}\geq 1- \beta_{1},
	\end{align*} 
	which concludes the proof.
\end{proof}

\begin{remark}
	As it can be observed, there is a  bilinearity between unknown variables $\kappa$ and $\alpha$ in $\Upsilon_2$. To resolve it, we consider $\alpha$ in a discrete set as $\alpha \in \{\alpha_1,\dots,\alpha_l\}$. The cardinality $l$ is then taken into account when determining the necessary amount of data to solve SOP, as shown in~\eqref{EQ:12}.
\end{remark}

\section{Data-Driven Guarantee for SBF Construction}
Here, we aim at constructing an SStF between each unknown subsystem and its finite MDP with a certified confidence level by establishing a probabilistic relation between optimal values of $\text{SOP}_\varsigma$ and $\text{ROP}$~\cite{esfahani2014performance}.

\begin{theorem}\label{Thm:6}
	Consider unknown dt-SCS $\Lambda$ in~\eqref{EQ:13}. Let $\Upsilon_1$ and $\Upsilon_2$ be Lipschitz continuous, with respect to $x$ and $(d,x)$ with Lipschitz constants, respectively, $\mathscr{L}_{1}, \mathscr{L}_{2_t}$, for given $\alpha_t$ where $ t\in\{1,\dots,l\}$, and any $\nu\in U$. Consider the $\text{SOP}_\varsigma$ in~\eqref{SCP1} with $\psi^*_\varsigma$, $\mathcal Q^* = [\gamma^*;\varpi^*;\mathcal Z^{11*};\mathcal Z^{12*};\mathcal Z^{22*};{\kappa}^*_{1};\dots;\kappa^*_z]$, and
	\begin{align}\label{EQ:12}
		N(\varepsilon_{2_t},\!\beta_2):=\min\!\Big\{N\in\N \,\big|\sum_{t=1}^{l}\sum_{i=0}^{c-1}\binom{N}{i}\varepsilon^i_{2_t}(1-\varepsilon_{2_t})^{N-i}\leq\beta_2\Big\},
	\end{align}
	where $\beta_2,\varepsilon_{2_t}\in [0,1]$ for any $t\in\{1,\dots,l\}$, with $c,l$ being, respectively, number of unknown variables in $\text{SOP}_\varsigma$, and cardinality of finite set of $\alpha$. If
	\begin{align}\label{Con1}
		\psi^*_\varsigma+\max_t\mathscr{L}_{\Upsilon_t}\eta^{-1}(\varepsilon_{2_t}) \leq 0,
	\end{align}	
	with $\mathscr{L}_{\Upsilon_t} := \max \big\{\mathscr{L}_{1},\mathscr{L}_{2_t}\big\}$ for any $t\in\{1,\dots,l\}$, and $\eta(r): \mathbb R_{\ge 0}\rightarrow [0,1]$, which depends on the geometry of $X\times D$ and the sampling distribution, then the data-driven $\mathcal S$ is an SStF between $\hat\Lambda$ and $\Lambda$, with a confidence of $1-\beta$ with $ \beta = \beta_1 + \beta_2$, \emph{i.e.,}
	\begin{align*}
		\PP^N\big\{\hat\Lambda\cong_{\mathcal{S}}\Lambda\big\}\ge 1-\beta_1 - \beta_2,
	\end{align*}
	where $\beta_1\in(0,1]$ is as in Lemma~\ref{Lem:2}. 
\end{theorem}

\begin{proof}
	According to~\cite[Theorem 4.3]{esfahani2014performance}, one can quantify the
	closeness between optimal values of
	ROP and SOP$_N$ as
	\begin{align}\label{EQ:12_1}
		\PP^N \Big\{0\leq\psi^*_R-\psi^*_N\leq\max_t\varepsilon_{1_t}\Big\}\geq 1-\beta_2,
	\end{align}
	with $$ N\big(\eta(\frac{\varepsilon_{1_t}}{\mathrm L_{\mathrm {SP}}\mathscr{L}_{\Upsilon_t}}),\beta_2\big),$$ where
	$\varepsilon_{1_t}\in [0,1]$, $\eta(s): \mathbb R_{\ge 0}\rightarrow [0,1]$, and $\mathrm{L}_{\mathrm{SP}}$ being a Slater point which is considered here as $1$ given that the original ROP in~\eqref{ROP} is a $\min$-$\max$ optimization program~\cite[Remark 3.5]{esfahani2014performance}.
	
	\noindent
	From~\eqref{EQ:12_1}, it can be concluded that $\psi^*_N\leq\psi^*_R\leq\psi^*_N+\max_t\varepsilon_{1_t}$ with a confidence of $1-\beta_2$. One also has $\psi^*_{N} \leq \psi^*_{\varsigma}$ with a confidence of $1-\beta_{1}$ according to Lemma~\ref{Lem:2}. Consequently, $\psi^*_{R}\leq\psi^*_{N} + \max_t\varepsilon_{1_t} \leq \psi^*_{\varsigma} + \max_t\varepsilon_{1_t}.$ If $\psi^*_{\varsigma} + \max_t\varepsilon_{1_t} \leq 0$, it implies that $\psi^*_{R} \leq 0$. By defining events $\mathcal A_1 :=\{ \psi^*_{N} \leq \psi^*_{\varsigma}\}$ and $\mathcal A_2 :=\{\psi^*_{R}\leq\psi^*_{N} + \max_t\varepsilon_{1_t}\}$, where $\PP\big\{\mathcal A_1\big\}\geq1-\beta_{1}$ and $\PP\big\{\mathcal A_2\big\}\geq1-\beta_2$, the concurrent occurrence of  events $\mathcal A_1$ and $\mathcal A_2$ can be computed as:
	\begin{align}\label{EQ:19}
		&\PP\big\{\mathcal A_1\cap \mathcal A_2\big\}=1-\PP\big\{\bar {\mathcal A_1}\cup \bar {\mathcal A_2}\big\},
	\end{align}
	where $\bar {\mathcal A_1}$ and $\bar {\mathcal A_2}$ are the complement of $\mathcal A_1$ and $\mathcal A_2$, respectively. Since
	\begin{align*}
		&\PP\big\{\bar {\mathcal A_1}\cup \bar {\mathcal A_2}\big\}\leq\PP\big\{\bar {\mathcal A_1}\big\}+\PP\big\{\bar {\mathcal A_2}\big\},
	\end{align*}
	and by leveraging \eqref{EQ:19}, one can readily conclude that 
	\begin{align}
		\PP\big\{\mathcal A_1\cap \mathcal A_2\big\}\geq 1-\PP\big\{\bar {\mathcal A_1}\big\}-\PP\big\{\bar {\mathcal A_2}\big\}\nonumber
		\geq 1-\beta_1 - \beta_{2}.
	\end{align}
	Since $\varepsilon_{2_t}=\eta(\frac{\varepsilon_{1_t}}{\mathscr{L}_{\Upsilon_t}})$~\cite{esfahani2014performance}, one has $\varepsilon_{1_t}= \mathscr{L}_{\Upsilon_t}\eta^{-1}(\varepsilon_{2_t})$. Then one can recast the condition $\psi^*_{\varsigma} + \max_t\varepsilon_{1_t} \leq 0$ as $\psi^*_\varsigma+\max_t\mathscr{L}_{\Upsilon_t}\eta^{-1}(\varepsilon_{2_t}) \leq 0$. Hence, if $\psi^*_\varsigma+\max_t\mathscr{L}_{\Upsilon_t}\eta^{-1}(\varepsilon_{2_t}) \leq 0$, then the constructed $\mathcal S$ from $\text{SOP}_\varsigma$ in~\eqref{SCP1} is an SStF between $\hat\Lambda$ and $\Lambda$ with a confidence of $1-\beta_1 - \beta_2$, which completes the proof.
\end{proof}

In the next lemma, we compute the function $\eta$ which is required for checking condition~\eqref{Con1}.

\begin{lemma}\label{Function_g}
	The function $\eta$ in~\eqref{Con1} fulfills the following condition~\cite[Proposition 3.8]{esfahani2014performance}:  
	\begin{align}\label{New}
		\eta(r) \leq \PP \big [\mathbb B_r(x,d)\big],\quad\quad \forall r\in\mathbb R_{\ge 0}, \forall (x,d) \in X\times D,
	\end{align}
	with $\mathbb B_r(c) \subset X\times D$ being an open ball with center $c$ and radius $r$. By gathering data from an $(n+p)$-dimensional \emph{hyper-rectangle} uncertainty set $X\times D$ with a \emph{uniform} distribution, the function $\eta$ in~\eqref{New} is then quantified as 
	\begin{align}\notag
		\eta(r) &=\frac{\text{Vol}(\mathbb B_r(x,d))}{2^{n+p}\text{Vol}(X\times D)} = \frac{\frac{\pi^{\frac{{n+p}}{2}}}{\Gamma(\frac{{n+p}}{2} + 1)}r^{n+p}}{2^{n+p}\text{Vol}(X\times D)}\\\label{g-function}
		& = \frac{\pi^{\frac{{n+p}}{2}}r^{n+p}}{2^{n+p}\Gamma(\frac{{n+p}}{2} + 1)\text{Vol}(X\times D)}, 
	\end{align}
	with $\text{Vol}(\cdot)$ and $\Gamma$ being volume set and Gamma function, respectively. 
\end{lemma}

To assess condition~\eqref{Con1}, it is necessary to determine $\mathscr{L}_{\Upsilon_t}$. The following lemmas present computations of $\mathscr{L}_{\Upsilon_t}$ for both linear and nonlinear stochastic systems

\begin{lemma}\label{Lem:3}
	Given a linear dt-SCS $x(k+1)=Ax(k) + B\nu(k)+ Ed(k) + \varsigma(k)$, let $(x-\hat x)^\top P(x-\hat x)$ be an SStF with a positive-definite matrix $P\in\mathbb{R}^{n\times n}$. Then $\mathscr{L}_{\Upsilon_t}$ is computed as $\mathscr{L}_{\Upsilon_t} = \max \big\{\mathscr{L}_{1},\mathscr{L}_{2_t}\big\}$, with
	\begin{align*}
		\mathscr{L}_{1}& = 4s_1 (\lambda_{\min}(P) + \lambda_{\max}(P)),\\
		\mathscr{L}_{2_t}& = 2\lambda_{\max}(P) (2\mathcal Y_1^2 s_1 + 2\mathcal Y_1 \mathcal Y_2 s_2 + 2\mathcal Y_1 \mathcal Y_3 s_3 + \mathcal Y_1\rho + \mathcal Y_3\rho \\
		&~~~~+ 2\mathcal Y_3^2 s_3 + 2\mathcal Y_3 \mathcal Y_2 s_2 + 2\mathcal Y_3 \mathcal Y_1 s_1 + 2s_1\alpha_t) + 2s_4s_5,
	\end{align*}
	where
	$\Vert A\Vert \leq \mathcal Y_1$, $\Vert B\Vert \leq \mathcal Y_2$, $\Vert E\Vert \leq \mathcal Y_3$, $\Vert x\Vert \leq s_1$ for any $x\in X$, $\Vert \nu\Vert \leq s_2$ for any $\nu\in U$, $\Vert d\Vert \leq s_3$ for any $d\in D$, $\Vert [d-\hat d;x-\hat x]\Vert \leq s_4$ for any $x\in X$, $\hat x\in \hat X$, $d\in D$, $\hat d\in \hat D$,  and $\Vert \mathcal Z \Vert = s_5$.
\end{lemma}	

\begin{figure*}
	\begin{align}\notag
	&\mathscr L_{{2_t}} = \max\limits_{x\in X, d\in D}\Vert \begin{bmatrix}
	2((Ax + B\nu+ Ed) - \mathcal P_x(A\hat x\!+\! B\nu\!+\! E\hat d))^\top PA\!-\! 2\alpha_t(x\!-\!\hat x)^\top P\\2((Ax \!+\! B\nu\!+\! Ed) \!-\! \mathcal P_x(A\hat x\!+\! B\nu\!+\! E\hat d))^\top PE
	\end{bmatrix}- 2 \begin{bmatrix}
	d-\hat d\\
	x-\hat x
	\end{bmatrix}^\top\begin{bmatrix}
	\mathcal Z^{11}&\mathcal Z^{12}\\
	\mathcal Z^{21}&\mathcal Z^{22}
	\end{bmatrix}\Vert\\\notag
	&\leq \max\limits_{x\in X, d\in D}\Big\{\Vert  2((Ax \!+\! B\nu\!+\! Ed) \!-\! \mathcal P_x(A\hat x\!+\! B\nu\!+\! E\hat d))^\top PA\Vert \!+\!\Vert 2\alpha_t(x\!-\!\hat x)^\top P\Vert \!+\! 2\Vert\begin{bmatrix}
	d-\hat d\\
	x-\hat x
	\end{bmatrix}\Vert ~\Vert\begin{bmatrix}
	\mathcal Z^{11}&\mathcal Z^{12}\\
	\mathcal Z^{21}&\mathcal Z^{22}
	\end{bmatrix}\Vert\\\notag
	&~~~+ \Vert  2((Ax + B\nu+ Ed) - \mathcal P_x(A\hat x+ B\nu+ E\hat d))^\top PE\Vert \Big\}\\\notag
	&\leq\max\limits_{x\in X, d\in D} \Big\{2\Vert P \Vert\big(\Vert A \Vert(\Vert A x\Vert + \Vert B  u \Vert + \Vert E d \Vert + \rho + \Vert A\hat x + B\nu + E\hat d\Vert)+\alpha_t(\Vert x \Vert + \Vert \hat x \Vert)\\\notag
	&~~~+ \Vert E \Vert(\Vert A x\Vert + \Vert B  u \Vert + \Vert E d \Vert + \rho + \Vert A\hat x + B\nu + E\hat d\Vert)\big)+ 2\Vert\begin{bmatrix}
	d-\hat d\\
	x-\hat x
	\end{bmatrix}\Vert~ \Vert\begin{bmatrix}
	\mathcal Z^{11}&\mathcal Z^{12}\\
	\mathcal Z^{21}&\mathcal Z^{22}
	\end{bmatrix}\Vert\Big\}\\\label{EQ:100}
	& \leq 2\lambda_{\max}(P) (2\mathcal Y_1^2 s_1 + 2\mathcal Y_1 \mathcal Y_2 s_2 + 2\mathcal Y_1 \mathcal Y_3 s_3 + \mathcal Y_1\rho + 2\mathcal Y_3^2 s_3 + 2\mathcal Y_3 \mathcal Y_2 s_2 + 2\mathcal Y_3 \mathcal Y_1 s_1 + \mathcal Y_3\rho + 2s_1\alpha_t) + 2s_4s_5.
	\end{align}
	\rule{\textwidth}{0.1pt}
\end{figure*}

\begin{proof}
	We first compute $\mathscr{L}_{1}$ and $\mathscr{L}_{2_t}$, and then take the maximum between them. By defining
	\begin{align}\notag
		\mathscr L_{{2_t}}:\left\{
		\hspace{-2mm}\begin{array}{l}\max\limits_{x\in X, d\in D}\Vert\frac{\partial \Upsilon_{2}}{\partial (d,x)}\Vert\\
			\,\,\,\,\,\,\, \text{s.t.} \quad \quad\!\!\Vert x\Vert \leq s_1, \Vert d\Vert \leq s_3, \Vert [d-\hat d;x-\hat x]\Vert \leq s_4,\end{array}\right.
	\end{align}
	one can reach the chain of inequalities in~\eqref{EQ:100}. For the calculation of $\mathscr{L}_{1}$, given that $\lambda_{\min}(P)\Vert x-\hat x \Vert^2 \leq (x-\hat x)^\top P(x-\hat x)$, it follows that $\gamma = \lambda_{\min}(P)$ in~\eqref{EQ:11}. Consequently, we obtain:
	\begin{align*} 
		\mathscr{L}_{1} &=\max\limits_{x\in X, \Vert x\Vert \leq s_1}\Vert 2\lambda_{\min}(P) (x - \hat x) - 2(x-\hat x)^\top P\Vert\\
		& \leq 4s_1 (\lambda_{\min}(P) + \lambda_{\max}(P)).
	\end{align*}
	Then $\mathscr{L}_{\Upsilon_t} = \max \big\{\mathscr{L}_{1},\mathscr{L}_{2_t}\big\} = \max\big\{4s_1 (\lambda_{\min}(P) + \lambda_{\max}(P)),2\lambda_{\max}(P) (2\mathcal Y_1^2 s_1 + 2\mathcal Y_1 \mathcal Y_2 s_2 + 2\mathcal Y_1 \mathcal Y_3 s_3 + \mathcal Y_1\rho + 2\mathcal Y_3^2 s_3 + 2\mathcal Y_3 \mathcal Y_2 s_2 + 2\mathcal Y_3 \mathcal Y_1 s_1 +\mathcal Y_3\rho + 2s_1\alpha_t) + 2s_4s_5\big\}$, which concludes the proof.
\end{proof}

We now compute $\mathscr{L}_{\Upsilon_t}$ for \emph{nonlinear} stochastic systems.

\begin{lemma}\label{Lem:3_1}
	Given a nonlinear dt-SCS $x(k+1)=f(x(k),\nu(k), d(k)) + \varsigma(k)$, let $(x-\hat x)^\top P(x-\hat x)$ be an SStF with a positive-definite matrix  $P\in\mathbb{R}^{n\times n}$. Then $\mathscr{L}_{\Upsilon_t}$ is computed as $\mathscr{L}_{\Upsilon_t} = \max \big\{\mathscr{L}_{1},\mathscr{L}_{2_t}\big\}$, with
	\begin{align*}
		\mathscr{L}_{1} &=4s_1 (\lambda_{\min}(P) + \lambda_{\max}(P)),\\
		\mathscr{L}_{2_t} &= 2\lambda_{\max}(P) (2\mathcal Y_f \mathcal Y_x + \mathcal Y_x\rho+ 2\mathcal Y_f \mathcal Y_d + \mathcal Y_d\rho + 2s_1\alpha_t) + 2s_4s_5,
	\end{align*} 
	where $\Vert f(x,\nu,d)\Vert \leq \mathcal Y_f$, $\Vert \partial_{x}f(x,\nu,d) \Vert \leq \mathcal Y_{x}$, $\Vert \partial_{d}f(x,\nu,d) \Vert\leq \mathcal Y_{d}$, $\Vert x\Vert \leq s_1$ for any $x\in X$, $\Vert [d-\hat d;x-\hat x]\Vert \leq s_4$ for any $x\in X$, $\hat x\in \hat X$, $d\in D$, $\hat d\in \hat D$,  and $\Vert \mathcal Z \Vert = s_5$.
\end{lemma}

\begin{figure*}
	\begin{align}\notag
	&\mathcal I_{{2_t}} = \max\limits_{x\in X, d\in D}\Vert \begin{bmatrix}
	2(f(x,\nu,d) - \mathcal P_x(f(\hat x,\nu,\hat d)))^\top P\partial_{x}f(x,\nu,d)- 2\alpha_t(x-\hat x)^\top P\\2(f(x,\nu,d) - \mathcal P_x(f(\hat x,\nu,\hat d)))^\top P\partial_{d}f(x,\nu,d)
	\end{bmatrix} - 2 \begin{bmatrix}
	d-\hat d\\
	x-\hat x
	\end{bmatrix}^\top\begin{bmatrix}
	\mathcal Z^{11}&\mathcal Z^{12}\\
	\mathcal Z^{21}&\mathcal Z^{22}
	\end{bmatrix}\Vert\\\notag
	& \leq  \max\limits_{x\in X, d\in D} \Big\{2\Vert P \Vert\big(\Vert \partial_{x}f(x,\nu,d) \Vert(\Vert f(x,\nu,d) \Vert + \Vert \mathcal P_x(f(\hat x,\nu,\hat d))\Vert) + \Vert \partial_{d}f(x,\nu,d) \Vert(\Vert f(x,\nu,d) \Vert + \Vert \mathcal P_x(f(\hat x,\nu,\hat d))\Vert)\\\notag
	&~~~+\alpha_t(\Vert x \Vert + \Vert \hat x \Vert)\big)+2\Vert\begin{bmatrix}
	d-\hat d\\
	x-\hat x
	\end{bmatrix}\Vert~ \Vert\begin{bmatrix}
	\mathcal Z^{11}&\mathcal Z^{12}\\
	\mathcal Z^{21}&\mathcal Z^{22}
	\end{bmatrix}\Vert\Big\}\\\notag
	&\leq\max\limits_{x\in X, d\in D} \Big\{2\Vert P \Vert\big(\Vert \partial_{x}f(x,\nu,d) \Vert(\Vert f(x,\nu,d) \Vert + \rho + \Vert f(\hat x,\nu,\hat d)\Vert) + \Vert \partial_{d}f(x,\nu,d) \Vert(\Vert f(x,\nu,d) \Vert + \rho + \Vert f(\hat x,\nu,\hat d)\Vert)\\\notag
	&~~~+\alpha_t(\Vert x \Vert + \Vert \hat x \Vert)\big)+2\Vert\begin{bmatrix}
	d-\hat d\\
	x-\hat x
	\end{bmatrix}\Vert ~\Vert\begin{bmatrix}
	\mathcal Z^{11}&\mathcal Z^{12}\\
	\mathcal Z^{21}&\mathcal Z^{22}
	\end{bmatrix}\Vert\Big\}\\\label{EQ:101}
	& \leq 2\lambda_{\max}(P) (2\mathcal Y_f \mathcal Y_x + \mathcal Y_x\rho + 2\mathcal Y_f \mathcal Y_d + \mathcal Y_d\rho + 2s_1\alpha_t) + 2s_4s_5.
	\end{align}
	\rule{\textwidth}{0.1pt}
\end{figure*}

\begin{proof}
	By defining
	\begin{align}\notag
		\mathscr L_{{2_t}}:\left\{
		\hspace{-2mm}\begin{array}{l}\max\limits_{x\in X, d\in D}\Vert\frac{\partial \Upsilon_{2}}{\partial (d,x)}\Vert\\
			\,\,\,\,\,\,\, \text{s.t.} \quad \quad\!\!\Vert x\Vert \leq s_1, \Vert d\Vert \leq s_3, \Vert [d-\hat d;x-\hat x]\Vert \leq s_4,\end{array}\right.
	\end{align}
	one can obtain the chain of inequalities in~\eqref{EQ:101}. For $\Upsilon_1$:
	\begin{align*} 
		\mathscr{L}_{\Upsilon_1} &=\max\limits_{x\in X, \Vert x\Vert \leq s_1}\Vert 2\lambda_{\min}(P) (x - \hat x) - 2(x-\hat x)^\top P\Vert\\
		& \leq 4h (\lambda_{\min}(P) + \lambda_{\max}(P)).
	\end{align*}
	Then $\mathscr{L}_{\Upsilon_t} = \max \big\{\mathscr{L}_{1},\mathscr{L}_{2_t}\big\} = \max\big\{4s_1 (\lambda_{\min}(P) + \lambda_{\max}(P)),2\lambda_{\max}(P) (2\mathcal Y_f \mathcal Y_x + \mathcal Y_x\rho + 2\mathcal Y_f \mathcal Y_d + \mathcal Y_d\rho + 2s_1\alpha_t) + 2s_4s_5\big\}$, which concludes the proof.
\end{proof}

\begin{figure*}
	\begin{align}\notag
		\EE &\Big [\mathcal V(\kappa,f(x,\nu,\varsigma),\hat{f}(\hat x,\nu,\varsigma))\big | x, \hat x, \nu\Big ] =  \EE\Big[\sum_{i=1}^M\Big[\mathcal S_i(\kappa_i,f_i(x_i,\nu_i,d_i,\varsigma_i),\hat{f}_i(\hat x_i, \nu_i,\hat d_i,\varsigma_i))\,|\,x_i,\hat x_i,\hat{\nu}_i,d_i,\hat d_i\Big]\Big]\\\notag
		&=\sum_{i=1}^M\EE\Big[\mathcal S_i(\kappa_i,f_i(x_i,\nu_i,d_i,\varsigma_i),\hat{f}_i(\hat x_i, \nu_i,\hat d_i,\varsigma_i))\,|\,x_i,\hat x_i,\hat{\nu}_i,d_i,\hat d_i\Big]\\\notag
		&\leq\sum_{i=1}^M\big(\alpha_i\mathcal S_i(\kappa_i,x_i,\hat{x}_i)+\varpi_i +\begin{bmatrix}
			d_i\!-\!\hat d_i\\
			x_i\!-\!\hat x_i
		\end{bmatrix}^\top\!\!\!\begin{bmatrix}
			\mathcal Z_i^{11}&\mathcal Z_i^{12}\\
			\mathcal Z_i^{21}&\mathcal Z_i^{22}
		\end{bmatrix}\!\!\begin{bmatrix}
			d_i\!-\!\hat d_i\\
			x_i\!-\!\hat x_i
		\end{bmatrix}\!\big)  \\\notag
		&=\sum_{i=1}^M\big (\alpha_i\mathcal S_i(\kappa_i,x_i,\hat{x}_i)+\varpi_i\big)+\begin{bmatrix}
			d_1- \hat d_1\vspace{-0.2cm}\\
			\vdots\vspace{-0.2cm}\\
			d_M- \hat d_M\vspace{-0.1cm}\\
			x_1- \hat x_1\vspace{-0.2cm}\\
			\vdots\vspace{-0.2cm}\\
			x_M- \hat x_M\\
		\end{bmatrix}^\top\begin{bmatrix}
			\mathcal Z_1^{11}\vspace{-0.1cm}&&&\mathcal Z_1^{12}\vspace{-0.1cm}&&\\
			&\ddots\vspace{-0.1cm}&&&\ddots\vspace{-0.1cm}&\\
			&&\mathcal Z_M^{11}&&&\mathcal Z_M^{12}\\
			\mathcal Z_1^{21}\vspace{-0.1cm}&&&\mathcal Z_1^{22}\vspace{-0.1cm}&&\\
			&\ddots\vspace{-0.1cm}&&&\ddots\vspace{-0.1cm}&\\
			&&\mathcal Z_M^{21}\vspace{-0.03cm}&&&\mathcal Z_M^{22}\vspace{-0.03cm}
		\end{bmatrix}\begin{bmatrix}
			d_1- \hat d_1\vspace{-0.2cm}\\
			\vdots\vspace{-0.2cm}\\
			d_M- \hat d_M\vspace{-0.1cm}\\
			x_1- \hat x_1\vspace{-0.2cm}\\
			\vdots\vspace{-0.2cm}\\
			x_M- \hat x_M\\
		\end{bmatrix}\\\notag
		&=\sum_{i=1}^M \big(\alpha_i\mathcal S_i(\kappa_i,x_i,\hat{x}_i)+\varpi_i\big)+\begin{bmatrix}
			x_1- \hat x_1\\
			\vdots\\
			x_M- \hat x_M\\
		\end{bmatrix}^\top\begin{bmatrix}
			\mathcal M\\
			\mathds{I}
		\end{bmatrix}^\top\begin{bmatrix}
			\mathcal Z_1^{11}\vspace{-0.1cm}&&&\mathcal Z_1^{12}\vspace{-0.1cm}&&\\
			&\ddots\vspace{-0.1cm}&&&\ddots\vspace{-0.1cm}&\\
			&&\mathcal Z_M^{11}&&&\mathcal Z_M^{12}\\
			\mathcal Z_1^{21}\vspace{-0.1cm}&&&\mathcal Z_1^{22}\vspace{-0.1cm}&&\\
			&\ddots\vspace{-0.1cm}&&&\ddots\vspace{-0.1cm}&\\
			&&\mathcal Z_M^{21}\vspace{-0.03cm}&&&\mathcal Z_M^{22}\vspace{-0.03cm}
		\end{bmatrix}\begin{bmatrix}
			\mathcal M\\
			\mathds{I}
		\end{bmatrix}\begin{bmatrix}
			x_1- \hat x_1\\
			\vdots\\
			x_M- \hat x_M\\
		\end{bmatrix}\\\label{Eq:17}
		&\le\sum_{i=1}^M \alpha_i\mathcal S_i(\kappa_i,x_i,\hat{x}_i) +\sum_{i=1}^M\varpi_i \leq \alpha \mathcal V(\kappa,x,\hat{x})+\varpi.
	\end{align}
	\rule{\textwidth}{0.1pt}
\end{figure*}

\subsection{Data-Driven Finite MDPs via Maximum Likelihood Estimation}\label{MLE}

Here, we construct finite MDPs from data by estimating parameters of the probability distribution via maximum likelihood estimation (MLE)~\cite{myung2003tutorial}. If the underlying stochasticity has a Gaussian distribution, its  mean and standard deviation can be estimated via MLE as
\begin{align*}
	\hat \mu_{\hat N} = \frac{1}{\hat N}\sum_{j = 1}^{\hat N} \tilde x_j, \quad \hat\sigma_{\hat N}^2 = \frac{1}{{\hat N}-1}\sum_{j = 1}^{\hat N} (\tilde x_j - \hat\mu_{\hat N})^2,
\end{align*}
where $\hat \mu_{\hat N},\hat \sigma_{\hat N}$ are the \emph{empirical} mean and standard deviation given $\hat N$ sampled data. Additionally, MLE approach can be used to estimate parameters of \emph{any arbitrary} probability distributions. We then use the estimated parameters from MLE method and construct a finite MDP via the results of Section~\ref{MDPs}. Although it is possible to provide an asymptotic confidence bound for MLE using Fisher information~\cite{le2012asymptotic}, we leave it to a future work for the sake of an easier presentation.

\section{Compositional Construction of SBF for Interconnected dt-SCS}\label{Sec:Com}
Here, we propose a compositional dissipativity approach to build an SBF for an interconnected network using SStF of individual subsystems. The constructed SBF is then utilized to compute the probabilistic mismatch between trajectories of the interconnected system $\Lambda$ and its finite MDP $\hat\Lambda$, as presented in Theorem~\ref{Thm:4}.

\begin{theorem}\label{Thm:3}
	Consider an interconnected dt-SCS $\Lambda = \mathcal{I}(\Lambda_1,\ldots,\Lambda_M)$ composed of $M\in\mathbb N_{0}^+$ subsystems~$\Lambda_i$. Let there exist an SStF between each subsystem $\Lambda_i$ and its finite MDP $\hat\Lambda_i$ with a confidence of $1-\beta$, with $ \beta = \beta_1 + \beta_2$, as in Theorem~\ref{Thm:6}. Then 
	\begin{equation}\label{Comp: SBF}
		\mathcal V(\kappa,x,\hat x) := \sum_{i=1}^{M}\mathcal S_i(\kappa_i,x_i,\hat x_i),
	\end{equation}
	is an SBF between $\hat\Lambda=\mathcal{I}(\hat\Lambda_1,\ldots, \hat\Lambda_{M})$ and $\Lambda=\mathcal{I}(\Lambda_1,\ldots, \Lambda_{M})$ with a confidence of $1-\sum_{i=1}^{M}\beta_{i}$, where $\beta_i= \beta_{1_i}-\beta_{2_i}$, 	
	if
	\begin{align}\label{Con_1}
		&\begin{bmatrix}
			\mathcal M\\\mathds{I}
		\end{bmatrix}^\top \mathcal Z_{cmp}\begin{bmatrix}
			\mathcal M\\\mathds{I}
		\end{bmatrix}\preceq 0, \\\notag
		\text{with} \quad 
		\mathcal Z_{cmp}:=&\begin{bmatrix}
			\mathcal Z_1^{11}\vspace{-0.1cm}&&&\mathcal Z_1^{12}\vspace{-0.1cm}&&\\
			&\ddots\vspace{-0.1cm}&&&\ddots\vspace{-0.1cm}&\\
			&&\mathcal Z_M^{11}&&&\mathcal Z_M^{12}\\
			\mathcal Z_1^{21}\vspace{-0.1cm}&&&\mathcal Z_1^{22}\vspace{-0.1cm}&&\\
			&\ddots\vspace{-0.1cm}&&&\ddots\vspace{-0.1cm}&\\
			&&\mathcal Z_M^{21}\vspace{-0.03cm}&&&\mathcal Z_M^{22}\vspace{-0.03cm}
		\end{bmatrix}\!\!.
	\end{align}
\end{theorem}
\vspace{0.2cm}
\begin{proof}
	We first show that SBF $\mathcal V$ in \eqref{Comp: SBF} fulfills condition \eqref{EQ:15}. For any $x=\intcc{x_1;\ldots;x_M}\in X$ and  $\hat x=\intcc{\hat x_1;\ldots;\hat x_M}\in \hat X$:
	\begin{align}\notag
		&\Vert x-\hat x \Vert^2\le\sum_{i=1}^M \Vert  x_i-\hat x_i \Vert^2
		\le \sum_{i=1}^M \frac{\mathcal S_i(\kappa_i,x_i, \hat x_i)}{\gamma_{i}}\le \bar\gamma\mathcal V(\kappa,x,\hat x),
	\end{align}
	with $\bar\gamma=\sum_{i=1}^M \frac{1}{\gamma_{i}}$. Hence, condition~\eqref{EQ:15} is met with $\gamma=\frac{1}{\bar\gamma}$. Now we continue with showing condition~\eqref{EQ:16}. By utilizing condition~\eqref{Con_1} and defining
	\begin{align}\notag
		\alpha r\Let \max\Big\{\sum_{i=1}^M\alpha_i r_i\,\,\big|\, r_i  {\ge 0},\,\,\sum_{i=1}^M r_i=r\Big\}, \quad\varpi\Let\sum_{i=1}^M\varpi_i,
	\end{align}
	the chain of inequalities in \eqref{Eq:17} can be acquired. Then condition \eqref{EQ:16} is also fulfilled. 
	
	\noindent
	We now show that the proposed $\mathcal V$ in~\eqref{Comp: SBF} is an SBF between $\hat \Lambda=\mathcal{I}(\hat\Lambda_1,\ldots, \hat\Lambda_{M})$ and $\Lambda=\mathcal{I}(\Lambda_1,\ldots, \Lambda_{M})$ with a confidence of $1-\sum_{i=1}^{M}\beta_{i}$, where $\beta_i= \beta_{1_i}-\beta_{2_i}$. By defining events $\mathcal A_i$ as $\mathcal A_i: \big\{\hat\Lambda_i\cong_{\mathcal{S}}\!\Lambda_i\big\}$ for all $i \in\{1,\dots,M\}$, where $\PP\big\{\mathcal A_i\big\}\ge 1-\beta_i$, the concurrent occurrence of events $\mathcal A_i$ can be quantified as:
	\begin{align}\label{eq:proof11}
		&\PP\big\{\mathcal A_1\cap\dots\cap \mathcal A_{M}\big\}=1-\PP\big\{\bar {\mathcal A}_1\cup\dots\cup \bar {\mathcal A}_{M} \big\},
	\end{align}
	where $\bar {\mathcal A}_i$ are complements of $\mathcal A_i, \forall i \in\{1,\dots,M\}$. Since
	\begin{align}\notag
		\PP&\big\{\bar {\mathcal A}_1\cup\dots\cup \bar {\mathcal A}_{M} \big\}\leq\PP\big\{\bar {\mathcal A}_1\big\}+\dots+\PP\big\{\bar {\mathcal A}_{M}\big\},
	\end{align}
	and by leveraging \eqref{eq:proof11}, one can finally conclude that 
	\begin{align}\label{eq:8}
		\PP\big\{\mathcal A_1\!\cap\!\dots\!\cap\! \mathcal A_{M}\big\}\!\geq\! 1-\PP\big\{\bar {\mathcal A}_1\big\}+\dots+\PP\big\{\bar {\mathcal A}_{M}\big\}
		\!\geq\! 1-\sum_{i =1}^{M} \beta_i.
	\end{align}
	Hence, $\mathcal V$ is an SBF between $\hat\Lambda$ and $\Lambda$ with a confidence of $1-\sum_{i=1}^{M}\beta_{i}$, with $\beta_i= \beta_{1_i}-\beta_{2_i}$, which concludes the proof.
\end{proof}

\section{Case Study: Room Temperature Network}\label{Sec:Case}

We showcase our data-driven results using a room temperature network consisting of $100$ rooms, each with unknown models, interconnected in a circular topology, and equipped with cooling systems. The temperature dynamics, denoted as $x(\cdot)$, can be described through the following interconnected network~\cite{meyer}:
\begin{align}\notag
	\Lambda\!:x(k+1)=Ax(k)+\theta T_{c}\nu(k)+ \digamma T_{E} + \varsigma(k),
\end{align}
where the matrix $A$ has diagonal entries $ a_{ii}=1-2\aleph-\digamma-\theta\nu_i(k)$, $i\in\{1,\ldots,M\}$, off-diagonal entries $a_{i,i+1}=a_{i+1,i}=a_{1,M}=a_{M,1}=\aleph$, $i\in \{1,\ldots,M-1\}$, and other entries being zero. Symbols $\aleph$, $\digamma$, and $\theta$ are thermal factors between rooms $i \pm 1$ and $i$, the outside environment and the room $i$, and the cooler and the room $i$, respectively.
In addition, $x(k)=[x_1(k);\ldots;x_{M}(k)]$, $x(k)=[\varsigma_1(k);\ldots;\varsigma_{M}(k)]$, $T_E=[T_{e_1};\ldots;T_{e_{M}}]$, with  $T_{e_i}=-1\,^\circ C$, $\forall i\in\{1,\ldots,{M}\}$, being the outside temperatures. The cooler temperature is $T_c=5\,^\circ C$ and the control input is $\nu\in\{0,0.05,0.1,0.15,0.2\}$. Now by characterizing each individual room as 
\begin{align}\label{sub_room}
	\Lambda_i\!:x_i(k+1)&=a_{ii}{x_i}(k)+\aleph (d_{i-1}(k) + d_{i+1}(k))+ \theta T_{c} \nu_i(k) +\digamma T_{e_i} + \varsigma_i(k),
\end{align}
where $d_0 = d_{M}, d_{M+1} = d_1$, one has $\Lambda=\mathcal{I}(\Lambda_1,\ldots,\Lambda_{M})$, with a coupling matrix $\mathcal M$ as $\bar m_{i,i+1}=\bar m_{i+1,i}=\bar m_{1,M}=\bar m_{M,1}=1$, $i\in \{1,\ldots,M-1\}$, and other entries being zero. We assume the model of each room is unknown to us. The main target is to compositionally construct a finite MDP as well as a data-driven SBF via solving $\text{SOP}$~\eqref{SCP1}. Accordingly, we utilize the data-driven finite MDP and synthesize controllers regulating the temperature of each room in a safe set $X_i = [-0.5,0.5]$ with a guaranteed probabilistic confidence.

We consider our SStF as $\mathcal S_i(\kappa_i,x_i,\hat x_i) = \kappa_{1_i}(x_{i} - \hat x_{i})^4 + \kappa_{2_i}(x_{i} - \hat x_{i})^2 + \kappa_{3_i}$. We also fix $\varepsilon_{t_i} = 0.025$, $\beta_{2_i} = 10^{-4}$, and $\rho_i = 0.05$, a-priori. According to~\eqref{EQ:12}, we compute $N_i = 911$ required for solving $\text{SOP}$ in~\eqref{SCP1}. We also fix $\mu_i = 0.1$, $\beta_{1_i} = 10^{-4}$ and compute $L_i = 643$ according to Lemma~\ref{Lem:2}. By solving $\text{SOP}$~\eqref{SCP1} with $N_i, L_i$, we obtain the corresponding decision variables as
\begin{align}\notag
	&\mathcal S_i(\kappa_i,x_i,\hat x_i) = 0.11(x_{i} - \hat x_{i})^4 + 0.14(x_{i} - \hat x_{i})^2 + 143,\\\notag
	& \mathcal Z^{11}_i = 0.001, \mathcal Z^{22}_i = -0.01, \mathcal Z^{12}_i = \mathcal Z^{21}_i = 0,\\\label{new}
	&\gamma_i^* = 141, \varpi^*_i = 0.42, \psi^*_{\varsigma_i}=-0.3019,
\end{align}
with a fixed $\alpha_i = 0.99$. We now compute $\mathscr{L}_{\Upsilon_{t_i}} = 0.8$ according to Lemma~\ref{Lem:3_1}. We also compute $\eta^{-1}(\varepsilon_{2_{t_i}})$ according to Lemma~\ref{Function_g} as $\eta^{-1}(\varepsilon_{2_{t_i}}) = 0.362$. Since $\psi^*_{\varsigma_i}+\max_t\mathscr{L}_{\mathcal H_{t_i}}\eta^{-1}(\varepsilon_{2_{t_i}}) = -11\times 10^{-3} \leq 0$, the constructed data-driven $\mathcal S_i$ is an SStF between each unknown room $\Lambda_i$ and its finite MDP $\hat\Lambda_i$, with a confidence of at least $1-\beta_{1_i} - \beta_{2_i} = 1 - 2\times 10^{-4}$.

We now construct an SBF for the interconnected rooms via SStF of individual rooms, constructed from data. By leveraging $\mathcal  Z_i$ as in~\eqref{new}, the matrix $\mathcal Z_{cmp}$ is reduced to
\begin{align*}
	\mathcal Z_{cmp}=\begin{bmatrix} 0.001\mathds{I}_{100} & 0 \\ 0 & -0.01\mathds{I}_{100} \end{bmatrix}\!\!,
\end{align*}
and compositionality condition \eqref{Con_1} is reduced to
\begin{align*}
	\begin{bmatrix} \mathcal M \\ \mathds{I}_{100} \end{bmatrix}^\top\!\!&\mathcal Z_{cmp}\begin{bmatrix} \mathcal M \\ \mathds{I}_{100} \end{bmatrix}=0.001\mathds{I}_{100}\mathcal M^\top\mathcal M-0.01\mathds{I}_{100} \preceq 0.
\end{align*}
Hence, one can certify that $\mathcal  V(\kappa, x,\hat x) = \sum_{i=1}^{100}\{\mathcal S_i(\kappa_i,x_i,\hat x_i)\} = \sum_{i=1}^{100}\{0.11(x_{i} - \hat x_{i})^4 + 0.14(x_{i} - \hat x_{i})^2 + 143\}$ is an SBF between the interconnected rooms $\Lambda$ and its finite MDP $\hat\Lambda$ with $\gamma = 14100,\alpha = 0.99, \varpi = 42,$ and a confidence of $1- \sum_{i=1}^{100}\beta_{1_i} -\sum_{i=1}^{100}\beta_{2_i}= 98\%$. Hence, by employing the results of Theorems~\ref{Thm:4} and~\ref{Thm:3}, we guarantee that the mismatch between state trajectories of $\Lambda$ and $\hat\Lambda$ remains within $\varepsilon = 0.5$ during $\mathcal T= 5$ ($45$ minutes) with a probability of $95\%$ and a confidence of $98\%$.

Let us now synthesize a controller for $\Lambda$ via its data-driven finite MDP $\hat \Lambda$, constructed via the MLE approach with $\hat N = 10^5$, such that the controller regulates state of each room within $[-0.5,0.5]$. To do so, we first synthesize a controller for each abstract room $\hat \Lambda_i$ via \texttt{AMYTISS} \cite{lavaei2020amytiss} and then refine it back over unknown original room $\Lambda_i$. Accordingly, the overall controller for the network would be a vector whose entries are controllers for individual rooms. Closed-loop trajectories of a representative room with several noise realizations are depicted in Fig.~\ref{Simulation}. As observed, all trajectories respect the safety specification.

\begin{figure}[h]
	\centering 
	\includegraphics[width=0.53\linewidth]{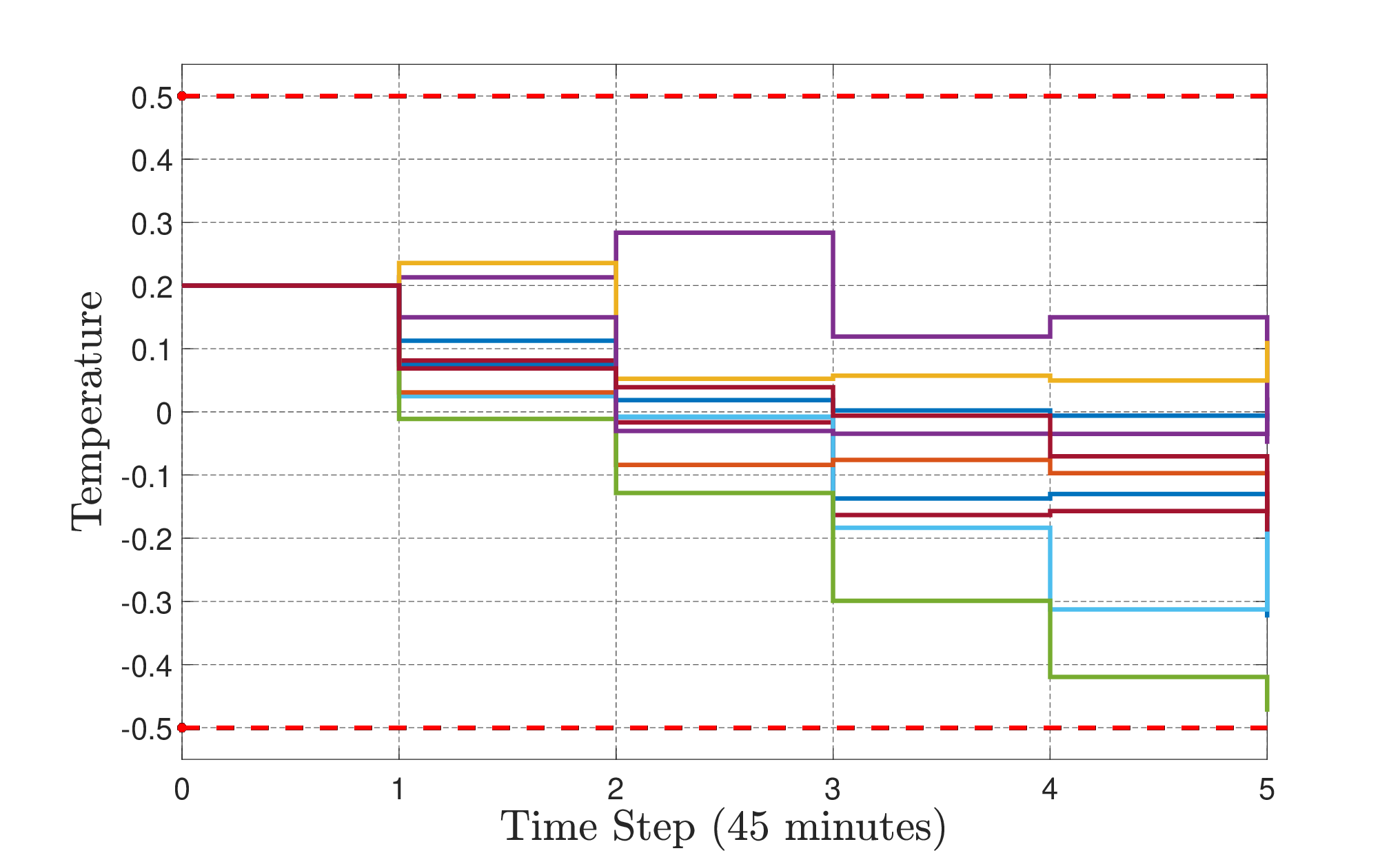}
	\caption{Closed-loop state trajectories of an unknown representative room with several noise realizations.}
	\label{Simulation}
\end{figure}

\section{Conclusion}

In this work, we developed a compositional data-driven technique using dissipativity reasoning for constructing finite MDPs for large-scale stochastic networks with unknown mathematical models. The main goal was to leverage stochastic bisimulation functions (SBF) and quantify the closeness between an unknown original network and its data-driven finite MDP, while proposing a certified probabilistic confidence. In our proposed scheme, we first constructed a stochastic storage function between each unknown subsystem and its data-driven finite MDP with an a-priori confidence level. We then provided dissipativity-type compositional conditions to construct an SBF for an unknown interconnected network using its data-driven SStF of subsystems. We verified our results over a room temperature network composing $100$ rooms with unknown dynamics.

\bibliographystyle{alpha}
\bibliography{biblio}

\end{document}